\documentclass[11pt]{article}

\usepackage{amsfonts}
\usepackage{amssymb}
\usepackage{amsthm}
\usepackage{amsmath}
\usepackage{algorithm}
\usepackage{algorithmic}
\usepackage{mathtools}
\usepackage{fullpage}

\usepackage{amsfonts}
\usepackage{amsmath} 
\usepackage{amssymb}
\usepackage{enumerate}

\newcommand{\ignore}[1]{}

\newcommand{\PROOF}{\begin{proof}}%

\newcommand{\QED}{\end{proof}}

\newcommand{\pair}[1]{\langle #1 \rangle}

\newcommand{\binary}{\{0,1\}^*}

\newcommand{\p}{{\mathrm{p}}}
\newcommand{\ptwo}{{\p_{\thinspace\negthinspace_2}}}

\newcommand{\pspace}{{\mathrm{pspace}}}

\newcommand{\co}[1]{\mathrm{co}#1}

\newcommand{\DTIME}{\mathrm{DTIME}}

\newcommand{\DTIMEF}{\mathrm{DTIMEF}}
\newcommand{\NP}{{\ensuremath{\mathrm{NP}}}} 
 
\newcommand{\coNP}{\co{\NP}}

\newcommand{\BPP}{{\rm BPP}}

\newcommand{\PSPACE}{{\rm PSPACE}}

\renewcommand{\P}{\ensuremath{{\mathrm P}}}

\newcommand{\EXP}{{\rm EXP}}

\newcommand{\SigmaEXPkPiEXPk}{\SigmaEXPkPiEXPk}

\newcommand{\PH}{{\rm PH}}

\newcommand{\T}{\mathrm{T}}

\newcommand{\leqp}{\leq^\p}
\newcommand{\leqpT}{\leqp_\T}

\renewcommand{\Pr}{\mathrm{P}_r}

\renewcommand{\leqp}{\leq^\P}
\usepackage{amsthm}

\newtheorem{theorem}{Theorem}[section]
\newtheorem{corollary}[theorem]{Corollary}
\newtheorem{lemma}[theorem]{Lemma}
\newtheorem{proposition}[theorem]{Proposition}

\theoremstyle{definition}

\newtheorem*{definition}{Definition}
\newtheorem*{defn}{Definition}

\newtheorem{example}[theorem]{Example}

\newtheorem*{example*}{Example}
\newtheorem*{examples*}{Examples}

\theoremstyle{remark}

\numberwithin{equation}{section}
\numberwithin{figure}{section}

\newcommand{\LA}{L_A}

\def\binary{\lbrace 0,1 \rbrace}

\def\binarystar{\lbrace 0,1 \rbrace ^ *}

\newcommand{\oraclesPNP}{\{A \mid \P^A = \NP^A\}}

\title{\bf Polynomial-Time Random Oracles and 
  Separating Complexity Classes}

\author{John M. Hitchcock 
\\ Department of Computer Science
  \\ University of Wyoming
\\ jhitchco@cs.uwyo.edu
\and 
Adewale Sekoni 
\\ Department of Computer Science
\\ University of Wyoming
\\ asekoni@uwyo.edu
\and Hadi Shafei
\\ Department of Mathematics and Computer Science
\\ Northern Michigan University
\\ hshafei@nmu.edu
}

\date{}

\begin{document}

\maketitle

\begin{abstract}

Bennett and Gill (1981) showed that $\P^A \neq \NP^A \neq \coNP^A$ for
a random oracle $A$, with probability 1. We investigate whether this
result extends to individual polynomial-time random oracles. We
consider two notions of random oracles: p-random oracles in the sense
of martingales and resource-bounded measure (Lutz, 1992; Ambos-Spies
et al., 1997), and p-betting-game random oracles using the betting
games generalization of resource-bounded measure (Buhrman et al.,
2000). Every p-betting-game random oracle is also p-random; whether
the two notions are equivalent is an open problem.

\begin{enumerate}[\upshape (1)]
\item\label{abstract:one} We first show that $\P^A \neq \NP^A$ for
  every oracle $A$ that is p-betting-game random.
\end{enumerate}
Ideally, we would extend (\ref{abstract:one}) to p-random oracles. We
show that answering this either way would imply an unrelativized
complexity class separation:
\begin{enumerate}[\upshape(1)]
\setcounter{enumi}{1}
\item If $\P^A \neq \NP^A$ relative to every p-random oracle $A$, then
  $\BPP \neq \EXP$.
\item If $\P^A = \NP^A$ relative to some p-random oracle $A$,
  then $\P \neq \PSPACE$.
\end{enumerate}

Rossman, Servedio, and Tan (2015) showed that the
polynomial-time hierarchy is infinite relative to a random oracle,
solving a longstanding open problem. We consider whether we can extend
(\ref{abstract:one}) to show that $\PH^A$ is infinite relative to
oracles $A$ that are p-betting-game random.
Showing that $\PH^A$ separates at even its first level
would also imply an unrelativized
complexity class separation:
\begin{enumerate}[\upshape (1)]
\setcounter{enumi}{3}
\item If $\NP^A \neq \coNP^A$ for a $\p$-betting-game measure 1 class
  of oracles $A$, then $\NP \neq \EXP$.
\item If $\PH^A$ is infinite relative to every $\p$-random oracle $A$, then $\PH
  \neq \EXP$.
\end{enumerate}

\end{abstract}

 \begin{section}{Introduction}
Bennett and Gill \cite{BennettGill81} initiated the study of random
oracles in computational complexity, proving that $\P^A \neq \NP^A$
for a random oracle $A$, with probability 1. Subsequent work showed
that this holds for {\em individual} random oracles. Book, Lutz, and
Wagner \cite{Lutz:BLW} showed that $\P^A \neq \NP^A$ for every oracle
$A$ that is algorithmically random in the sense of Martin-L{\"o}f
\cite{MartinLof66}.  Lutz and Schmidt \cite{Lutz:CSRPO} improved this
further to show $\P^A \neq \NP^A$ for every oracle $A$ that is
pspace-random \cite{Lutz:AEHNC}.

We investigate whether this extends to individual polynomial-time
random oracles \cite{Lutz:AEHNC,AmbMay97}. To show that $\P^A \neq
\NP^A$ for p-random oracles $A$, we need to show that if $\P^A =
\NP^A$, then there is a polynomial-time martingale that succeeds on
$A$. This means that if $A$ makes $\P^A = \NP^A$, then $A$ is somehow
predictable or simple.

Allender and Strauss \cite{AllStr94} proved that %
  $\{A \mid \P^A \neq \BPP^A \}$ has p-measure 0, which implies
  that $\P^A = \BPP^A$ for every p-random oracle $A$. This strengthens
  another result of Bennett and Gill \cite{BennettGill81} that $\P^A =
  \BPP^A$ holds for a random oracle $A$, with probability 1.  Allender
  and Strauss's proof relies on derandomization \cite{NisWig94} and is
  a different approach than Bennett and Gill.  For $\P$ vs $\NP$
  oracles, the best known is the pspace-randomness result of Lutz and
  Schmidt \cite{Lutz:CSRPO}.  In related work, Kautz and Miltersen
  \cite{KauMil96} showed that if $A$ is an algorithmically random
  oracle, then $\NP^A$ does not have p-measure 0.  Because the class
  $\oraclesPNP$ has Hausdorff dimension 1 \cite{Hitchcock:HDOC}, there
  is a fundamental limit to how strongly a martingale can succeed on
  the class.

Each oracle $A$ is associated with a {\em test language} $\LA$. This
language is tally and $0^n \in \LA$ if and only if in the $2^n$
tribes of $n$ strings following $0^n$, there is at least one tribe
contained in $A$.  (See Section \ref{sec:betting} for a precise
definition of $\LA$. Bennett and Gill used a slightly different, but
equivalent formulation of the test language.)  It is clear that $\LA
\in \NP^A$. From \cite{BennettGill81}, we know that $\{ A \mid \LA
\in \P^A \}$ has Lebesgue measure 0.  Since
$\P^A = \NP^A$ implies $\LA \in \P^A$, it follows that $\oraclesPNP$
has measure 0. We would like to show $\{ A \mid \LA \in \P^A \}$ has
p-measure 0.

Intuitively, if $\LA \in \P^A$, we would like to predict membership of
strings in $A$. This would be relatively simple if the $\P^A$ algorithm
asked only nonadaptive queries. However, since the queries may be
adaptive, there are potentially exponentially many queries -- too many
to be considered by a polynomial-time martingale.

The difficulty is martingales are forced to bet on strings in
lexicographic order. Buhrman et al. \cite{BvMRSS01} introduced an
extension of resource-bounded measure using {\em betting
  games}. Betting games are similar to martingales but they may
adaptively choose the order in which they bet on strings. Whether
betting games are equivalent to martingales is an open question
\cite{BvMRSS01}. The adaptiveness in betting games allows us to simulate
$\P^A$ algorithms. We show in Section \ref{sec:betting} that there
is a p-betting game succeeding on $\{ A \mid \LA \in
\P^A\}$. Therefore $\P^A \neq \NP^A$ for every p-betting-game random
oracle $A$.

 In Section \ref{sec:limitations}, we consider whether there are
 limitations to extending the betting games result. 
We show that determining whether or not
$\oraclesPNP$ has polynomial-time measure 0 (with respect to martingales)
 would imply a separation of complexity classes:
 \begin{itemize}
\item If $\oraclesPNP$ has $\p$-measure 0, then $\BPP \neq \EXP$.
\item If $\oraclesPNP$ does not have $\p$-measure 0, then $\P \neq \PSPACE$.
 \end{itemize}
This shows that determining the $\p$-measure of $\oraclesPNP$, or
resolving whether $\P^A \neq \NP^A$ for all $\p$-random $A$, is likely
beyond current techniques.

Bennett and Gill \cite{BennettGill81} also showed that $\NP^A \neq
\coNP^A$ for a random oracle $A$, with probability 1. Rossman,
Servedio, and Tan \cite{RoSeTa15} answered a longtime open question
\cite{Hastad86} by extending Bennett and Gill's result to separate every
level of the polynomial-time hierarchy. They proved an average case
depth hierarchy theorem for Boolean circuits which implies that the
polynomial-time hierarchy is infinite relative to a random oracle.
Can we show that $\PH$ is infinite relative to polynomial-time random
oracles as well? We show that extending our main result to separate
$\PH^A$ at even the first level would separate $\NP$ from
$\EXP$:
\begin{itemize}
\item If %
$\{A \mid \NP^A = \coNP^A\}$ has $\p$-betting-game measure 0, then
  $\NP \neq \EXP$.
\item If $\PH^A$ is infinite relative to every $\p$-random oracle $A$,
  then $\PH \neq \EXP$.
\end{itemize}

 \end{section}
 
 \begin{section}{Preliminaries}\label{sec:prelim}
We use standard notation. The binary alphabet is $\Sigma = \{0,1\}$,
the set of all binary strings is $\Sigma^*$, the set of all binary strings of  
length $n$ is $\Sigma^n$, and the set of all infinite binary sequences is
$\Sigma^\infty$. The empty string is denoted by $\lambda$.  We use
the standard enumeration of strings,
$s_0=\lambda,s_1=0,s_2=1,s_3=00,s_4=01,\ldots$, and the standard lexicographic  
ordering of strings corresponds to this enumeration. The characteristic 
sequence of a language $A$ is the sequence $\chi_A \in \Sigma^\infty$, where
$\chi_A[n] = 1 \iff s_n \in A$. We refer to $\chi_A[s_n] = \chi_A[n]$ as the 
characteristic bit of $s_n$ in $A$. A language $A$ can alternatively be seen 
as a subset of $\Sigma^*$, or as an element of $\Sigma^\infty$ via 
identification with its characteristic sequence $\chi_A$. Given strings $x,y$ 
we denote by $[x,y]$ the set of all strings $z$ such that $x\leq z \leq y$. 
For any string $s_n$ and number $k$, $s_n+k$ is the string $s_{n+k}$; e.g.
$\lambda + 4 = 01$. Similarly we denote by $A[x,y]$ the substring of the 
characteristic sequence $\chi_A$ that corresponds to the characteristic bits
of the strings in $[x,y]$.

\subsection{Martingales and Betting Games}

We now give a brief overview of martingales and betting games, and
how they are applied in computational complexity to define
resource-bounded measures and randomness notions. For further details,
we refer to \cite{Lutz:AEHNC,Lutz:QSET,AmbMay97,BvMRSS01,Harkins:ELBG}.

Betting games, which are also called nonmonotonic martingales,
originated in the field of algorithmic information theory.  In that
setting they yield the notion of Kolmogorov-Loveland randomness
(generalizing Kolmogorov-Loveland stochasticity)
\cite{Muchnik1998,MeMiNiStRe06}.  The concept was introduced to
computational complexity by Buhrman et al.~\cite{BvMRSS01}.
First, we recall the definition of a martingale:

\begin{definition}  A {\em martingale} is a function $d : \Sigma^* \to
  [0,\infty)$ such that for all $w \in \Sigma^*$, we have the
  following averaging condition:  
$$d(w) = \frac{d(w0)+d(w1)}{2}.$$
\end{definition}

Intuitively, a martingale is betting in order on the
characteristic sequence of an unknown language.  The martingale
starts with finite initial capital $d(\lambda)$.  The quantity $d(w)$
represents the current capital the martingale has after betting on
the first $|w|$ bits of a sequence that begins with $w$.  The
quantities $\pi(w,0) = d(w0)/2d(w)$ and $\pi(w,1) = d(w1)/2d(w)$
represent the fraction of its current capital that the martingale is
wagering on $0$ and $1$, respectively, being the next bit of the
sequence.  This next bit is revealed and the martingale has
$d(w0)=2 \pi(w,0) d(w)$ in the case of a 0 and $d(w1) = 2 \pi(w,1) d(w)$ in
the case of a 1.

Betting games are a generalization of martingales and have the
additional capability of selecting which position in a sequence, or
equivalently, which string in a language, to bet upon next.  A betting
game is permitted to select strings in a nonmonotone order, that is,
it may bet on longer strings, then shorter strings, then longer
strings again (with the important restriction that it may not bet on
the same string twice). Like martingales, betting games must also
satisfy the averaging law, i.e. the average of the betting game's
capital after betting on a string $s$ when $s$ belongs and when $s$
doesn't belong to the language is the same as its capital before
betting on $s$. We use the following definition of a betting game from 
\cite{BvMRSS01}.

\begin{definition}%
A betting game $G$ is an oracle Turing machine that maintains a
``capital tape'' and a ``bet tape,'' in addition to its standard query
tape and worktapes. The game works in rounds $i = 1, 2, 3, \ldots$ as
follows. At the beginning of each round $i$, the capital tape holds a
nonnegative rational number $C_{i-1}$. The initial capital $C_0$ is
some positive rational number. $G$ computes a query string $x_i$ to
bet on, a bet amount $B_i, 0 \leq B_i \leq C_{i-1}$, and a bet sign
$b_i\in\{-1,+1\}$. The computation is legal so long as $x_i$ does not
belong to the set $\{ x_1, \cdots , x_{i-1}\}$ of strings queried in
earlier rounds. $G$ ends round $i$ by entering a special query
state. For a given oracle language $A$, if $x_i\in A$ and $b_i =+1$,
or if $x_i \not\in A$ and $b_i=-1$, then the new capital is given by
$C_i := C_{i-1} + B_i$, else by $C_i := C_{i-1}-B_i$. We charge $M$
for the time required to write the numerator and denominator of the
new capital $C_i$ down. The query and bet tapes are blanked, and G
proceeds to round $i+1$.
\end{definition}

It is easy to see from the above definition that $b_i$ and $B_i$ can easily be 
computed from the current capital $C_i := C_{i-1}+b_iB_i$ of the betting game.
Therefore, we can equivalently define a betting game by describing the 
computation of the current capital $C_i$ without explicitly specifying the 
computation of $b_i$ and $B_i$. We do this because it is clearer and more 
intuitive to describe the computation of the current capital of the betting 
game presented in the next section.

\begin{definition}[]
If a betting game $G$ earns unbounded capital on a language $A$ (in
the sense that for every constant $c$ there is a point at which the capital
exceeds $c$ when betting on $A$), we say that $G$ {\em succeeds on} $A$.  The 
{\em success set} of a betting game $G$, denoted $S^\infty[G]$, is the set of 
all languages on which $G$ succeeds.  A betting game $G$ {\em succeeds on} a 
class $X$ of languages if $X \subseteq S^\infty[G]$.
\end{definition}

By adding a resource bound $\Delta$ on the computation of a betting
game or martingale, we get notions of resource-bounded measure on
$\Sigma^\infty$. For this paper the resource bounds we use are
$\p=\DTIMEF(n^{O(1)}),$ $\ptwo=\DTIMEF(2^{(\lg n)^{O(1)}})$, and
$\pspace=\mathrm{DSPACEF}(n^{O(1)})$.  We say a class $X \subseteq
\Sigma^\infty$ has $\Delta$-betting-game measure 0, if there is a
$\Delta$-computable betting game that succeeds on every language in
it. It has $\Delta$-measure 0 if the betting game is also a martingale
\cite{Lutz:AEHNC}. A class $X$ has $\Delta$-betting-game measure 1 if
$X^c$ has $\Delta$-betting-game measure 0. Similarly, $X$ has
$\Delta$-measure 1 if $X^c$ has $\Delta$-measure 0. A language $A$ is
$\Delta$-betting-game random if there is no $\Delta$-computable
betting game that succeeds on $A$. Similarly, $A$ is $\Delta$-random
if there is no $\Delta$-computable martingale that succeeds on $A$.

The ability of the betting game to examine a sequence nonmonotonically
makes determining its running time complicated, since each language
can induce a unique computation of the betting game.  In other words,
the betting game may choose to examine strings in different orders
depending upon the language it is wagering against.  Buhrman et
al.~looked at a betting game as an infinite process on a language,
rather than a finite process on a string.  They used the following
definition:

\begin{definition}[]
  A betting game $G$ runs in time $t(2^n)$ if for all languages $A$,
  every query of length $n$ made by $G$ occurs in the first $t(2^n)$
  steps of the computation.
\end{definition}

Specifically, once a $t(2^n)$-time-bounded betting game uses $t(2^n)$
computational steps, it cannot go back and select any
string of length $n$. Most importantly, no polynomial-time betting game can 
succeed on the class %
$\EXP =\DTIME(2^{n^{O(1)}})$.

\subsection{Martingales and Betting Games: Intuitive view}
Intuitively, a betting game can be viewed as the strategy of a gambler who bets
on infinite sequence of strings. The gambler starts with initial capital $C$, 
then begins to query strings to bet on.
The gambler's goal is to grow the capital 
$C$ without bound. The same view holds for martingales with the restriction 
that the gambler must bet on the strings in the standard ordering.

 \end{section}

\section{Betting Game Random Oracles}\label{sec:betting}

In this section we show that $\P^A \neq \NP^A$ for every $\p$-betting-game 
random oracle.

\begin{theorem}
\label{EbettingGameThm}\label{th:betting}
The class $\{A \ |\ \P^A \neq \NP^A\}$ has $\p$-betting-game measure 1. In
particular, $\P^A \neq \NP^A$ for every $\p$-betting-game random
  oracle $A$.
\end{theorem}
\begin{proof}
Given a language $A$ we define the test language
\[\LA=\{0^{n}\ |\ \textrm{Tribes}_{2^n,n}(A[0^n+1,0^n+n2^n]) = 1\},\]
where $\textrm{Tribes}_{2^n,n}:\binary^{n2^n}\longrightarrow\binary$ is defined 
as follows. Given $w\in\binary^{n2^n}$, first we view $w$ as concatenation of 
$2^n$ length $n$ strings
$w_1,w_2,\cdots,w_{2^n}$; i.e. $w = w_1w_2\cdots w_{2^n}.$
$\textrm{Tribes}_{2^n,n}(w)$ is $1$ if and only if $w_i = 1^n$ for
some $i$.
Secondly, we view $w$ as the substring $A[0^n+1,0^n+n2^n]$ of 
the characteristic sequence of some language $A$. With both views in mind, we 
define a tribe to be the set of strings whose characteristic bits are encoded 
by some $w_i$. For example, given any $i\in[1,2^n]$, the set of strings
$[0^n+(i-1)n+1,0^n+in]$ is a tribe because its characteristic bits are encoded 
by $w_i$. Since the $n$ strings in any tribe have length $O(n)$, an $\NP$ 
oracle machine can easily verify the membership of any $0^n$, therefore
$\LA \in \NP^A$. Now we define a betting game $G$ that succeeds on the set
$X=\{A \ |\ \P^A = \NP^A\}$, thereby proving the theorem. Our betting game $G$ 
is going to simulate oracle Turing machines on some strings in the set
$\{0^{n}\ |\ n\in\mathbb{N}\}$. Let $M_1, M_2, \cdots$ be an enumeration of all 
oracle TMs, where $M_i$ runs in time at most $n^{\lg i}+i$ on inputs of length 
$n$. The initial capital of $G$ is $2$ and we view it as composed of infinite 
``shares'' $a_i = b_i = 2^{-i}, i \in \mathbb{N}$ that are used by $G$ to bet on 
some of the strings it queries.
 
Before we go into the details of the implementation of $G$, we give a high
level view. The strategy of $G$ to succeed on $X$ is quite simple. For any
language $A$, the cardinality of $\{0^n \ |\ 0^n \notin \LA\}$ is either 
finite or infinite. When it is finite, after seeing a finite number of strings 
all following strings will  belong to $\LA$. $G$ uses ``shares'' $a_i$ 
reserved at its initialization to bet in this situation. On the other hand when 
it is infinite and $A \in X$ we can find an oracle TM $M_i$ that decides
$\LA$. Most importantly this TM rejects its input infinitely often and it is 
only in this situation that we bet with the $b_i$ ``shares''. Details follow.
  
First we specify the order in which $G$ queries strings followed by which 
strings it bets on. $G$ operates sequentially in stages $1, 2,\cdots$. In stage
$j$, $G$ queries $0^{n_j}$, where $n_j$ is the smallest integer such that all 
the strings queried in stage $j-1$ have length  less than $n_j$. $G$ then runs 
the oracle TM $M_{i+1}$ on $0^{n_j}$, where $i$ is the number of TMs simulated 
in the previous stages whose output was inconsistent with $\LA$ in one of the 
previous stages. During the simulation of $M_{i+1}$, $G$ answers any queries 
made by the TM either by looking up the string from its history, or if the string 
isn't in its history, then $G$ queries it. After the simulation, $G$ queries in 
the standard lexicographic order all the strings in the $2^{n_j}$ tribes that 
follow $0^{n_j}$ that haven't already been queried. Finally, to complete stage 
$j$, $G$ queries all the remaining strings of length at most the length of the 
longest string queried by $G$ so far.
  
Now we specify which strings $G$ bets on and how it bets with the $a_i$'s and 
$b_i$'s. In stage $j$, let $i$ and $n_j$ be such that $M_i$ is the Turing 
machine simulated in this stage and $0^{n_j}$ is the input it will
be simulated on. The only strings $G$ bets on will be the $n_j2^{n_j}$ strings 
following $0^{n_j}$; i.e. the tribes. We use $a_l$ and $b_i$, two of the 
infinite ``shares'' of our initial capital reserved by $G$ for betting, where 
$l$ is the smallest positive integer such that $a_l\neq0$. As will be shown 
later we do this because $G$ loses all of $a_l$ whenever $0^n \not\in \LA$.
The ``shares'' $a_l$ and $b_i$ are dynamic and may have their values updated as 
we bet with them. Therefore, the current capital of $G$ after each bet is
$\sum_{i=1}^{\infty}(a_i+b_i)$. Though we describe separately how $G$ bets with 
$a_l$ and $b_i$, we may bet with both simultaneously. We bet with some $a_l$ for 
every stage, but with the $b_i$'s we bet only when the output of the simulated TM 
is $0$. Therefore every time we bet with $b_i$ we also simultaneously bet with 
$a_l$. First let us see how $G$ bets in stage $j$ using the $a_l$ and then with 
$b_i$.
  
\textbf{Betting with ${a_l}$:} Our choice of $l$ ensures that
$a_l \neq 0$. In fact, $a_l$ will either increase, or reduce to $0$
after betting. If we lose $a_l$ in the current stage, then we use
$a_{l+1} = 2^{-(l+1)}$ to bet in the next stage. $G$ uses $a_l$ to bet that at 
least one of the $2^{n_j}$ tribes that follow $0^{n_j}$ is completely contained 
in $A$; i.e. $0^{n_j}\in \LA$. Call this event $\mathcal{B}_{n_j}$. It is easy 
to see that for sufficiently large $n_j$, when strings are included 
independently in $A$ with probability $1/2$, the probability of event
$\mathcal{B}_{n_j}$ is 
$$\Pr(\mathcal{B}_{n_j}) = 1-(1-2^{-n_j})^{2^{n_j}} \approx 1-1/e.$$ $G$ bets in 
such a way that whenever the sequence of strings seen satisfies the event
$\mathcal{B}_{n_j}$, $a_l$ increases by a factor of approximately $1/(1-1/e)$. 
If the sequence of strings does not satisfy event $\mathcal{B}_{n_j}$ then $G$ 
loses all of $a_l$ and will bet with $a_{l+1}$ in the next stage.
   
We now elaborate on how $a_l$ increases by a factor of approximately
$1/(1-1/e)$ when event $\mathcal{B}_{n_j}$ occurs. Let
$\omega \in \{0,1,\star\}^{n_j2^{n_j}}$ represent the current status of 
strings in $[0^{n_j}+1,0^{n_j}+n_j2^{n_j}]$, $\omega[i]$ indicates the 
status of string $0^{n_j}+i$, $\star$ indicates the string has not been queried 
by $G$ yet, and bits $0$ and $1$ have their usual meaning. Define
\[G_{a_l}(\omega) = \frac{a_l}{\textrm{Pr}(\mathcal{B}_{n_j})}\textrm{Pr}(\mathcal{B}_{n_j}|\omega),\]
where $\textrm{Pr}(\mathcal{B}_{n_j})$ is the probability a random language
satisfies event $\mathcal{B}_{n_j}$, and
$\textrm{Pr}(\mathcal{B}_{n_j}|\omega)$ 
is the conditional probability of the event $\mathcal{B}_{n_j}$ given the 
current status of the strings as encoded by $\omega$, i.e. given the strings in 
$[0^{n_j}+1, 0^{n_j}+n_j2^{n_j}]$ whose membership in $A$ has already been 
revealed, what is the probability that randomly assigning membership to other 
strings causes event $\mathcal{B}_{n_j}$ to occur. This probability is rational 
and easy to compute in $O(2^{2n})$ time by examining the status of the strings 
in each of the $2^{n}$ tribes in $[0^{n}+1, 0^{n}+n2^{n}]$. $G_{a_l}$ is 
essentially a martingale. Whenever the membership of any string in
$[0^{n_j}+1,0^{n_j}+n_j2^{n_j}]$ is revealed, $a_l$ is then updated to
$G_{a_l}(\omega)$. Given $\omega \in \{0,1,\star\} ^{n_j2^{n_j}}$ and
$b \in \{0,1,\star\}$, let $\omega^{i\rightarrow b}$ denote $\omega$ with its 
$i^\mathrm{th}$ symbol set to $b$. It is easy to see that
\[G_{a_l}(\omega^{i\rightarrow\star})=\frac{G_{a_l}(\omega^{i\rightarrow0})+G_{a_l}(\omega^{i\rightarrow1})}{2}.\]

For all sufficiently 
large $n_j$, $$G_{a_l}(\omega) = \frac{a_l}{\Pr(B_{n_j})} \approx a_l/(1-1/e)$$ for any string
$\omega \in \binary^{n_j2^{n_j}}$ that satisfies event $\mathcal{B}_{n_j}$ 
and $0$ for those that do not satisfy $\mathcal{B}_{n_j}$. It is important to 
note that $G$ can always bet with $a_l$ no matter the order in which it 
requests the strings in $[0^{n_j}+1,0^{n_j}+n_j2^{n_j}]$ that it bets on. But 
as will be shown next the ordering of these strings is important when betting 
with $b_i$.

\textbf{Betting with ${b_i}$:}
Finally, we specify how $G$ bets with ``share'' $b_i$ which is reserved for 
betting with $M_i$. $G$ only bets with $b_i$ when the simulation of $M_i$ on 
$0^{n_j}$ returns $0$. In this situation $G$ bets that at least $2^{n_j}-
(n_j^{\lg i}+i)$ tribes of the $2^{n_j}$ tribes that follow
$0^{n_j}$ are not contained in $A$. For simplicity, $G$ does not bet on the 
tribes that $M_i$ queried. We denote by $\mathcal{C}_{n_j}$ the event that all 
the tribes not queried by $M_i$ are not contained in $A$. Event
$\mathcal{C}_{n_j}$ occurs with probability at most
$(1-2^{-n_j})^{2^{n_j}-(n_j^{\lg i}+i)}\approx 1/e$, and is almost the 
complement of $\mathcal{B}_{n_j}$. In this case $G$ bets similarly to how it 
bets with $a_l$ and increases $b_i$ by a factor of
$1/\Pr(\mathcal{C}_{n_j}) \approx e$ whenever 
the sequence of strings that follow $0^{n_j}$ satisfies $\mathcal{C}_{n_j}$. If 
the sequence does not satisfy $\mathcal{C}_{n_j}$ then $G$ loses all of $b_i$.
  
We now argue that $G$ succeeds on $X$. Suppose $A \in X$ and $S \subseteq 
0^*$ is the set of input strings $G$ simulates on some TMs in stages 
$1, 2, \ldots$. Then there are two possibilities:
  
\begin{enumerate}
\item Finitely many strings in $S$ do not belong to $\LA$,
\item Infinitely many strings in $S$ do not belong to $\LA$.
\end{enumerate}
  
Denote by $s_k$ the $k^\textrm{th}$ string in $S$. In the first case, there 
must be a $k$ such that for every stage $j \geq k$, $s_j \in \LA$. Once 
we reach stage $k$, $G$ uses a ``share'' of its capital $a_i \neq 0$ to bet 
on $s_j$ belonging to $\LA$ for all $j \geq k$. Therefore $G$ will increase 
$a_i$ by a factor of approximately $1/(1-1/e)$ for all but finitely many 
stages $j \geq k$. Therefore the capital of $G$ will grow without bound in 
this case.
  
In the second case, we must reach some stage $k$ at which we use the correct
oracle TM $M_i$ that decides $\LA$ on inputs in $S$. From this stage onward 
$G$ will never change the TM it simulates on the strings in $S$ we
have not seen yet. In this case we are guaranteed this simulation will output 
$0$ infinitely often. It follows by the correctness of $M_i$ and the definition 
of $G$ that whenever the output of $M_i$ is $0$ the ``share'' of the capital
$b_i$ reserved for betting on $M_i$ will be increased by a factor of 
approximately $e$. Since this condition is met infinitely often, it follows 
that the capital of $G$ increases without bound in this case also.

Finally we show that $G$ can be implemented as a $O(2^{2n})$-betting game; i.e. 
after $O(2^{2n})$ time, $G$ will have queried all strings of length $n$. 
First, we bound the runtime of each round of the betting game; i.e. the time 
required to bet on a string. This should not be confused with the stages of $G$ 
which include several rounds of querying. In each round, we have to compute
$\sum_{i=1}^{\infty}(a_i+b_i)$ the current capital of $G$. This sum can
easily be computed in $O(2^n)$ time. This is because for each round we change
at most two ``shares'' $a_l$ and $b_i$ to some rational numbers that can be 
computed in $O(2^n)$ time. The remaining $a$ and $b$ ``shares'' with 
indices less than $i$ and $l$ respectively have values $0$ and those with 
indices grater than $i$ and $l$ respectively retain their initial values. We 
may also simulate a TM in each round. Since each simulated TM $M_i$ 
has $i \leq n$ it takes $O(n^{\lg n})$ time for the simulation of $M_i$ on
$0^n$. Therefore, each round is completed in $O(2^n)$ time. After the 
simulation $G$ requests all the remaining strings 
in $[0^n+1,0^n+n2^n]$ that were not queried during the simulation. Therefore, 
it takes $O(2^{2n})$ time for $G$ to have requested all strings of length $n$. 
\end{proof}

\section{Limitations}\label{sec:limitations}

In this section we examine the possibility of extending
Theorem \ref{th:betting}. We show that it cannot be improved to
$\p$-random oracles or improved to separate the polynomial-time hierarchy
without separating $\BPP$ or $\NP$ from $\EXP$, respectively. On the
other hand, showing that Theorem \ref{th:betting} cannot be improved
to $\p$-random oracles would separate $\PSPACE$ from $\P$.

\subsection{Does $\P^A \neq \NP^A$ for every $\p$-random oracle $A$?}

We showed in Theorem \ref{th:betting} that $\P^A \neq \NP^A$ for a
$\p$-betting-game random oracle. It is unknown whether $\p$-betting
games and $\p$-martingales are equivalent. If they are, then $\BPP
\neq \EXP$ \cite{BvMRSS01}. This is based on the following theorem and
the result that $\leqpT$-complete languages for $\EXP$ have
$\p$-betting-game measure 0 \cite{BvMRSS01}.

 \begin{theorem}[Buhrman et al. \cite{BvMRSS01}]
 \label{BppExpThm}
  If the class of $\leqpT$-complete languages for $\EXP$ has
  $\ptwo$-measure zero then $\BPP \neq \EXP$.
 \end{theorem}

We show improving Theorem \ref{th:betting} to
$\p$-random oracles would also imply $\BPP \neq \EXP$. First, we prove
the following for $\ptwo$-measure.
 
\begin{theorem}\label{th:P-NP-measure-limitation}
 If $\{A \mid \P^A \neq \NP^A\}$ has $\ptwo$-measure 1, then $\BPP \neq \EXP$.
 \end{theorem}
 \begin{proof}
  If $L$ is any $\leqpT$-complete language for $\EXP$, then 
  \[\NP^L \subseteq \EXP \subseteq \P^L \subseteq \NP^L.\] Therefore
  the class of $\leqpT$-complete languages for $\EXP$ is a subset of $\{A \mid \P^A
  = \NP^A\}$. If $\{A \mid \P^A = \NP^A\}$ has $\ptwo$-measure 0 then so
      does the class of $\leqpT$-complete languages of $\EXP$. Theorem
      \ref{BppExpThm} implies that $\BPP \neq \EXP$.
 \end{proof}
 
 \noindent
 We have the following for $\p$-random oracles by the
 universality of $\ptwo$-measure for $\p$-measure \cite{Lutz:AEHNC}.

\begin{corollary}\label{co:P-NP-random-limitation} If $\P^A \neq \NP^A$ for every $\p$-random oracle
  $A$, then $\BPP \neq \EXP$.
\end{corollary} 
\begin{proof}
The hypothesis implies that every $A$ with $\P^A = \NP^A$ is not
$\p$-random, i.e. there is a $\p$-martingale that succeeds on
$A$. Let $d'$ be a $\ptwo$-martingale $d'$ that is universal for all
$\p$-martingales \cite{Lutz:AEHNC}: $S^\infty[d] \subseteq
S^\infty[d']$ for every $\p$-martingale $d$. Then $d'$ succeeds
on $\{A \mid \P^A = \NP^A\}$.
\end{proof}

\subsection{Is it possible that $\P^A = \NP^A$ for some $\p$-random
  oracle $A$?}

Given Theorem \ref{th:P-NP-measure-limitation}, we consider the
possibility of whether $\oraclesPNP$ does not have $\p$-measure
0. Because Lutz and Schmidt \cite{Lutz:CSRPO} showed that this class
has $\pspace$-measure 0, it turns out that if it does not have
$\p$-measure 0, then we have a separation of $\PSPACE$ from $\P$.

\begin{theorem}[Lutz and Schmidt \cite{Lutz:CSRPO}]\label{th:pspace-P-NP}
The class $\oraclesPNP$ has $\pspace$-measure 0.
\end{theorem}
\noindent
We note that because every $\p$-betting game may be simulated by a
$\pspace$-martingale \cite{BvMRSS01}, Theorem \ref{th:pspace-P-NP}
follows as a corollary to Theorem \ref{th:betting}.
 \begin{lemma}
 \label{PspaceLemma}
  If $\P = \PSPACE$, then for every $\pspace$-martingale $d$,
  there is a $\p$-martingale $d'$ with $S^\infty[d] \subseteq
  S^\infty[d']$.
 \end{lemma}
 \begin{proof}
  Let $d:\binarystar \longrightarrow [0,\infty)$ be a
    $\pspace$-martingale.  Without loss of generality also assume that
    $d$ is exactly computable \cite{Lutz:WCE} and its output is in
    $\binary^{\leq p(n)},$ for some polynomial $p$. Consider the
    language $L_d = \{\pair{w,i,b}\ |\textrm{ the $i$th bit of $d(w)$
      is $b$}\}$.  Clearly $L_d \in \PSPACE$ and hence also in
    $\P$ by our hypothesis.  We can therefore compute $d(w)$ is
    polynomial time using $L_d$.
 \end{proof}
 
 \begin{theorem}
  If %
$\oraclesPNP$ does not have $\p$-measure 0, then $\P \neq \PSPACE$.
 \end{theorem}
 \begin{proof}
  Assume $\P = \PSPACE.$ Theorem \ref{th:pspace-P-NP} and Lemma
  \ref{PspaceLemma} imply that $\oraclesPNP$ has $\p$-measure 0.
 \end{proof}

\begin{corollary}\label{co:P-NP-random-limitation-2} If there is a $\p$-random oracle $A$ such that $\P^A
  = \NP^A$, then $\P \neq \PSPACE$.
\end{corollary}
 
\subsection{Is $\PH$ infinite relative to $\p$-betting-game random oracles?}

Bennett and Gill \cite{BennettGill81} showed that $\NP^A \neq \coNP^A$
for a random oracle $A$, with probability 1. Thus $\PH^A$ does not
collapse to its first level. Rossman, Servedio, and Tan
\cite{RoSeTa15} showed that $\PH^A$ is infinite relative to a random
oracle, with probability 1.

Can we improve Theorem \ref{th:betting} to show that $\PH^A$ does not
collapse for a $\p$-betting-game random oracle? This also has
complexity class separation consequences:

 \begin{theorem}
 \label{th:SigmaKThm}
  For $k > 0$, let $X_k = \{A \mid \sum_k^{\P,A} =
  \prod_k^{\P,A}\}$. If
  $X_k$ has $\ptwo$-betting-game measure zero, then $\sum_k^\P \neq \EXP$.
 \end{theorem}
 \begin{proof}
  We prove the contrapositive. Suppose $\sum_k^\P = \EXP$, then
  $\prod_k^\P = \EXP$. Given $A \in \EXP$, then the following containments
  hold:
$$ \Sigma_k^\P \subseteq \Sigma_k^{\P,A} \subseteq \EXP = \Pi_k^\P
   \subseteq \Pi_k^{\P,A} \subseteq \EXP = \Sigma_k^\P.$$
   turn implies that $\EXP \subseteq X_k$. Since $\EXP$ does not have
   $\ptwo$-betting-game measure zero \cite{BvMRSS01} then neither does
   $X_k$. Hence, the Theorem follows.
 \end{proof}

In particular, we have the following for the first level of PH:
\begin{corollary}\label{co:NP-coNP} If $\{ A \mid
\NP^A \neq \coNP^A \}$ has p-betting-game measure 1, then $\NP \neq
\EXP$.
\end{corollary}

Because it is open whether betting games have a union lemma
\cite{BvMRSS01}, it is not clear whether Corollary \ref{co:NP-coNP}
may be extended to show that if $\NP^A \neq \coNP^A$ for every
p-betting-game random oracle $A$, then $\NP \neq \EXP$. This extension
would hold if there is a $\ptwo$-betting game that is universal for
all $\p$-betting games. However, we do have the following for
$\p$-random oracles.

\begin{corollary} If $\NP^A \neq \coNP^A$ for every $\p$-random oracle
  $A$, then $\NP \neq \EXP$.
  \end{corollary}

\begin{corollary} If $\PH^A$ is infinite for every $\p$-random oracle
  $A$, then $\PH \neq \EXP$.
\end{corollary}

\section{Conclusion}

We have shown that $\P^A \neq \NP^A$ for every p-betting-game random
oracle $A$ (Theorem \ref{th:betting}). Establishing whether this also
holds for p-random oracles would imply either $\BPP \neq \EXP$
(Corollary \ref{co:P-NP-random-limitation}) or $\P
\neq \PSPACE$ (Corollary \ref{co:P-NP-random-limitation-2}). These
results, together with Theorems \ref{th:pspace-P-NP} and
\ref{th:SigmaKThm}, motivate investigating the status of $\PH$ relative
to pspace-random oracles. In particular:

\begin{enumerate}
\item Does $\{ A \mid \NP^A = \coNP^A \}$ have $\pspace$-measure 0?
\item More generally, does $\{ A \mid \PH^A \textrm{ collapses} \}$
  have $\pspace$-measure 0?
\end{enumerate}

 \bibliographystyle{plain}

\end{document}